\documentclass[proceedings,submission]{dmtcs}
\usepackage{amsmath}
\usepackage{amssymb}
\usepackage{amscd}
\usepackage{epic}
\usepackage{eepic}
\usepackage{ecltree}
\usepackage{epsfig,psfrag,cancel}
\usepackage[latin1]{inputenc}

\def \ca{\cancel}
\def \ALGO{\textrm{ALGO}}

\usepackage[usenames]{color}
\def\Expo{{\bf Expo}}
\RCSdef$Revision: 1.3 $\endRCSdef
\rcsMajMin
\def \app#1#2#3#4#5{\begin{array}{rccl} #1:&#2&\longrightarrow&#3\\ &#4&\longmapsto&#5\end{array}}
\author{Jean-Fran{\c c}ois Marckert\addressmark{1}
  \and Nasser Saheb-Djahromi\addressmark{1}%\thanks{But he is!}
  \and Akka Zemmari\addressmark{1}}
\title[Election with delays in trees]{Election algorithms with random delays in trees}
\address{\addressmark{1}LaBRI, Universit\'e de Bordeaux - CNRS, 351 cours de la
  Lib\'eration, 33405 Talence, France}
\revision{\rcsMaj}
\begin{document}
\newtheorem{proposition}{Proposition}
\newtheorem{theorem}{Theorem}
\newtheorem{lemma}{Lemma}
\newtheorem{corollary}{Corollary}
\newtheorem{remark}{Remark}
\newtheorem{Example}{Example}

\maketitle
\begin{abstract}
The election is a classical problem in distributed algorithmic. It aims to design and to analyze a distributed algorithm choosing a node in a graph, here, in a tree. In this paper, a class of randomized algorithms for the election is studied. The election amounts to removing leaves one by one until the tree is reduced to a unique node which is then elected. The algorithm assigns to each leaf a probability distribution (that may depends on the information transmitted by the eliminated nodes) used by the leaf to generate its remaining random lifetime. In the general case, the probability of each node to be elected is given. For two categories of algorithms, close formulas are provided.
\par
\vspace*{0.2cm}
\end{abstract}

\section{Introduction}
\subsection{The problem}
Starting from a configuration where all processors
are in the same state, the goal of an election algorithm is to obtain a configuration where
exactly one processor is in the state \it leader\rm, the other ones
being in the state \it lost\rm.  The (leader) election problem is
often the first problem to solve in a distributed environment. A
leader permits to centralize some information, to make some decisions,
to coordinate the processors for subsequent
tasks. Hence, the election problem -- first posed by Le Lann in
\cite{L77} -- is one of the most studied problems in distributed
algorithmic, and this under many different assumptions
\cite{T00}. The graph encoding the relations between the processors can be a ring, a tree, a complete or a
general connected graph. The system can be synchronous or asynchronous
and processors may have access to a total or partial information of
the geometry of the underlying graph, or of the current state of the
system, etc.\par

In this paper we consider the case of election in trees, when the
nodes have at time $t = 0$ a very partial information on the geometry
of the tree: each node only knows its number of neighbors. A
possible method for electing in a tree, introduced by Angluin
(\cite{A80} Theorem 4.4), amounts to eliminating successively the
leaves till only one node remains, the leader. In this
paper, we investigate this method in the general case: assume that a
node $u$ being a leaf (was a leaf at time $t=0$, or that becomes a leaf
at time $t$) decides to live a random remaining time $D_u$ before
being eliminated; in other words, it is eliminated at time $t+D_u$
except if it is elected before this date. Starting with a given tree $T_0$ at time 0, denote by $T_t$
the tree constituted with the non-eliminated nodes at time $t$. The
family $(T_t)_{t\geq 0}$ is a random process taking its values in the
set of trees. Given $T_0$, the distribution of $(T_t)_{t\geq 0}$ --
and then, also the probability that a given node is elected -- depends
on the way the nodes choose the distribution according to which they
will compute their random remaining lifetime. \\
\noindent -- In \cite{MS94} the authors consider two elementary
approaches. The first one is based on the assumption that \emph{all
sequences} of leaves elimination have the \emph{same} probability 
(no distributed algorithm seems to have this property).  Their second approach assumes that at
each step \emph{all leaves} have the same probability of being
removed. This corresponds to the case where the $D_u's$ are all
exponentially distributed with parameter 1. The authors study
thoroughly both approaches and prove many properties of 
resulting random processes. \\ 
-- In \cite{MSZ05}, the authors show
that if the nodes suitably choose their remaining random lifetime
then a totally fair election process is possible, the nodes 
being elected equally likely (in
Section \ref{exa} this example is revisited). In \cite{HSZ05a} and \cite{HRSZu}, the authors extend the
result from \cite{MSZ05} to a more general class of graphs: the
polyominoid graphs. They also prove a conjecture: the expected value
of the election duration is equal to $\log n$.\\

In this paper, we investigate the general case, namely, we consider
the case where a leaf $u$ generates its remaining lifetime $D_u$
according to a distribution ${\cal D}_u$, where ${\cal D}_u$ may
depend on all the information that $u$ has at its disposal (see
Remark \ref{rem1} below). We warm the reader to distinguish the notation ${\cal D}_u$ and $D_u$.

\begin{remark}\label{fr}
-- In order to avoid that two nodes may disappear exactly at the same
time, the distributions ${\cal D}_u$ need to avoid atoms (points with a positive mass).
Even if not recalled in
the statements, we assume that the distributions ${\cal D}_u$ have 
no atom. (In Section \ref{treo} a case where ${\cal D}_u$ maybe 0 with a positive probability arises
and leads to problems).

-- It is assumed throughout the paper, that the nodes own independent
random generators. This assumption is needed each time that the
independence argument is used in the paper.
\end{remark}

\subsection{The general scheme} 
Throughout this paper $T=(V,E)$ is a tree in the graph theoretic sense:
$V$ is its set of nodes, $E$ the set of edges. The graph $T$ is acyclic and connected, and undirected. The \it size \rm of $T$, denoted by $|T|$, is the
number of nodes.\par In the class of algorithms we study, a node $u$
becoming a leaf at time $t$ (or which was a leaf at time $t=0$)
disappears at time $t+D_u$ (except if it is elected before!); the
quantity $D_u$, called the remaining lifetime of $u$, is computed
locally by the leaf $u$. The description of the way $u$ chooses the
distribution ${\cal D}_u$ is crucial: this description is in fact
equivalent to the description of an algorithm using the general method
of elimination of leaves. We then enter into details here.\par

When a leaf is eliminated, it may transmit to its unique neighbor some
\sl information \rm (this notion will be formalized below). During the
execution of the algorithm, as a result of the successive eliminations
of the leaves, each internal node $u$ eventually becomes a leaf, say at
time $t_u$. At this time, it may use the information received to compute
the distribution ${\cal D}_u$: then, it generates a random variable $D_u$ following ${\cal D}_u$ using a random generator.
After this delay (at time $t_u+D_u$), $u$ is eliminated: it may transmit some information to its (unique) neighbor, and disappears from the tree. The
election goes on till eventually only one single node remains; this
node is then elected.\par
\begin{figure}[htbp]
\centerline{\includegraphics[height=2.2cm]{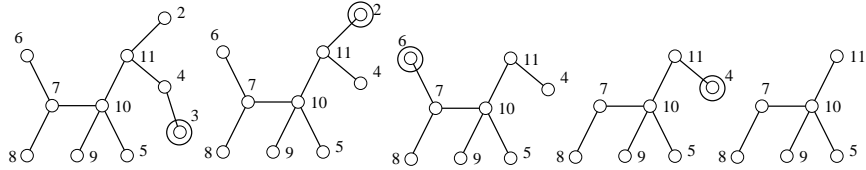}}
\label{el}
\caption{On this example, are circled at each step the next leaf to
disappear. On this example, the remaining lifetime of the leaf 11,
according to an algorithm $\Delta$ is allowed to depend on the
information given by the nodes 2 and 4; the information provided by 4 may include the information it received from the node 3. The total information received by 11 has a forest structure (a forest having 2,4 as roots, and having as set of nodes $2,3,4$, and possibly containing all the lifetimes, prescribed weights, and computed values of these nodes).} 
\end{figure}

As said above, the key point here is to understand that an algorithm (from the class we study) is parametrized by the way a node $u$ chooses -- according to the
information it has -- the distribution ${\cal D}_u$.

We here formalize more precisely what we understand by \sl information
received and information transmitted, \rm this needed to be coherent with the
distributed model we consider. This will straightforwardly leads to
the formal definition of our class of algorithms.
\begin{itemize}
\item[$a)$] The only information a node $u$ has at time 0 is its
degree $\deg_u$ and a prescribed weight $w_u$, which is an element of
$\mathbb{R}$, $\mathbb{R}^d$ or any set (this may be viewed as a personal parameter),
\item[$b)$] at its time of
disappearance a leaf $u$ transmits to its unique neighbor $v$ all the
information it has:\\ 
-- the information it has received from
its neighbors eliminated nodes,\\ 
-- the 4-tuple $L_u=(\deg_u,D_u,w_u,\Gamma_u)$ which is the \it local value \rm of $u$; the quantity $\Gamma_u$ is computed by $u$ using the information it has received and possibly the pair $(\deg_u,w_u)$. In the application we have, $\Gamma_u$ is used to compute ${\cal D}_u$, and then we assume that $\Gamma_u$ is not a function of $D_u$. We call $\Gamma_u$ the \it computed value \rm of $u$, it may belong to any set. See the remark below.
\end{itemize}
Assume that a node $u$ becomes a leaf at time $t$ when $k$ of its
$k+1$ neighbors $v_1,\dots, v_k$, have been
eliminated. Denote by $I_1,\dots,I_k$ the information these nodes have
transmitted to $u$. The node $u$ has at its disposal the multiset
$\{I_1,\dots,I_k\}$. Recursively, one sees that the structure of the information received by $u$ is a forest with $k$ rooted trees (a forest being here a multiset of trees) rooted at the $v_i$'s and constituted with eliminated nodes; this forest has the geometry of the tree $T$ fringed at the $v_i$'s. The node $u$ formally knows the local value of each of the nodes of this forest.
\begin{remark} \label{rem1}
$\bullet$ $w_u$ and $\Gamma_u$ are not used by each
   algorithm: when not used, they may be supposed to be 0.\par

\noindent $\bullet$ The notion of computed values aims to
   simplify the description of some algorithms, summing the needed information. Formally the
   transmission of this value is not necessary since it can be computed by a node having in hand all the other information.\par

\noindent $\bullet$ Let $\mu$ be a distribution on $\mathbb{R}$ with
cumulative distribution function $F$. If $U$ is uniform on $[0,1]$ 
then the law of $F^{-1}(U)$ is $\mu$, where $F^{-1}(u)=\inf\{x\mid F(x)\geq
u\}$ is the right continuous inverse of $F$; hence to simulate
any  distribution $\mu$, a uniform random
variable on $[0,1]$ is sufficient. We assume that the nodes have at
their disposal some independent random generators providing uniform
random values on $[0,1]$.
\end{remark}

Hence clearly, the information a node has received can be encoded
without loss of information by a labelled forest $f$, where each node
$v$ is labelled by the 4-tuple $L_v$. The set of received information
will then be identified with ${\cal F}$ the set of forests labelled by
4-tuple corresponding to the $L_u$'s.\par The other information at the
disposal of a given node $u$ that may be used to compute ${\cal D}_u$
is its own \it local information \rm
$L^\star_u=(\deg(u),w_u,\Gamma_u)$, where as said above $\Gamma_u$ has
been computed using $(\deg(u), w_u)$ and the received information. We
denote by ${\cal L}^\star$ the set of local information.\par

An algorithm is then just parametrized by a function $\Delta$
\[\app{\Delta}{{\cal F}\times {\cal L}^\star}{\cal M}{(f,l^\star)}{\Delta(f,l^\star)}\]
where ${\cal M}$ is the set of probability measures having their
support included in $[0,+\infty)$. The function $\Delta$ associates
with a pair $(f,l^\star)$ a probability distribution
$\Delta(f,l^\star)$.  Any map $\Delta$ encodes an algorithm
$\ALGO(\Delta)$: when $\ALGO(\Delta)$ is used, a node $u$ becoming a
leaf and having received the information $f$ and having as local information
$l_u^\star$, computes ${\cal D}_u=\Delta(f,l^\star)$ and generates $D_u$
according to ${\cal D}_u$. The maps $\Delta$ exemplified below depend
only on a part of the information received. The algorithms
$\ALGO(\Delta)$ are in the class of algorithms using the method of
Angluin, and satisfy the constraints to be distributed.
\begin{Example}\rm\label{esa1}We translate into the form $\ALGO(\Delta)$ the
algorithm defined in M\'etivier \& al. \cite{MSZ05}. For each node $u$, $w_u=1$. A node which is a leaf at time 0 computes $\Gamma_u=1$. Let $u$
be an internal node and $\Gamma_{v_1},\dots,\Gamma_{v_k}$ be the
computed values of the eliminated neighbors of $u$. Then $u$ computes:
\begin{equation}
\Gamma_u=1+\Gamma_{v_1}+\dots+\Gamma_{v_k}.
\end{equation}
Now the application $\Delta$ depends only on the computed values:
suppose that $u$ has received $(f,l^\star)$ and has computed
$\Gamma_u$, then ${\cal D}_u=\Delta(f,l^\star)$ is simply
$\Expo(\Gamma_u)$, the exponential distribution\footnote{a random variable r.v. ${\mathcal E}$ has
the distribution $\Expo(a)$, for some $a>0$ if
$\mathbb{P}({\mathcal E}\geq x)= \exp(-ax), \ \ \mbox{for all } x
\geq 0.$} with parameter
$\Gamma_u$. Hence, ${\cal D}_u=\Expo(1)$ if $u$ is a leaf at time
$0$, and if $u$ becomes a leaf later, then ${\cal
D}_u=\Expo(\Gamma_u)$, where $\Gamma_u$ equals one plus the size of the
forest of eliminated nodes leading to it (see Fig. 1). It
turns out that in this case, each node is elected equally likely (for all tree $T$). We provide in Section
\ref{exa} a new proof of this fact. M\'etivier et al. \cite{MSZ05},
\cite{HSZ05a} and \cite{HSZ05b} introduced election algorithms on
trees, $k$-trees and polyominoids having also this property.
\end{Example}

\medskip

%Even if formally the algorithms may depends
%on all the information transmitted, in the examples treated below, the
%mapping $\Delta$ will use only a part of this information: this
%will be encoded by a pair $(\Gamma_u,g_u)$.
%\[\Delta(f,l)=\Delta(\Gamma_u,g_u).\]

We address the question to compute according a general $\ALGO(\Delta)$,
the probability $q_u$ that a given node $u$ is eventually elected. In Section \ref{PGE}
we answer in the general case to this question, and express the result in terms of 
properties of some variables
arising in a related problem of directed elimination.

%A very related
%question is to find a function $\Delta$ allowing to target a
%distribution $(q_u)_{u\in V}$: if one wants the node to be elected
%proportionally to some prescribed weight $w_u$, is it possible to
%design $\Delta$ such that after the election using $\ALGO(\Delta)$,
%the probability that $u$ is elected is proportional to $w_u$?\par

In the sequel, we introduce and study two categories of
algorithms in the class of algorithms $\ALGO(\Delta)$.  Before
discussing their properties, we have to say that in order to get 
close formulas for $(q_u)_{u \in V}$, some stabilities in the
computations are necessary, and this is not possible for general
functions $\Delta$. The two categories
we propose raise on two different kinds of stability: the
$(\max,+)$ algebra in distribution, and the stable distributions for
the convolutions.  \par -- The first one is built using the properties
of the exponential distribution, and generalizes the computation of
M\'etivier \& al: the application $\Delta$ takes its values in the
set of exponential distributions union the set of convolutions of such
distributions. This category contains an algorithm $\ALGO(\Delta)$
such that $(q_u)_{u \in V}$ is proportional to the prescribed weights
$(w_u)_{u \in V}$. For technical reasons the prescribed weights
$(w_u)_{u \in V}$ are to be integer valued. When the $(w_u)_{u \in V}$
are allowed to be real numbers, we propose an algorithm which elects
proportionally to these weights in case of success, but which fails
with a \it low \rm probability,\par -- the second category may be less
interesting from an algorithmic point of view, since the algorithms
are more time consuming than the algorithms of the first category; it
has however two main advantages: it clarify in some sense the
properties needed to make the computation for a given function
$\Delta$, and it leads to a surprising proof of some mathematical
identities involving the function $\arctan$.

\section{General case: probability of a given node to be elected}
\label{PGE}
In this section, we give a general formula giving $(q_u)_{u \in V}$
for $\ALGO(\Delta)$. The proposition below is a generalization of a
proposition of M\'etivier \&. al \cite{MSZ05} (the coupling argument
we use is new).\par

The idea of the proof is to decompose the event $\{u \textrm{ is not
elected }\}$ into disjoint events: if $u$ is not elected, this means
that $u$ has become a leaf (or was a leaf at $t=0$) and then has been
eliminated. Let $t$ be the time when $u$ has become a leaf. At this time
$u$ had only one neighbor $v$, and since afterward $u$ was not
elected, this means that $u$ has disappeared before $v$. If at time 0,
$u$ has $k$ neighbors $v_1,\dots,v_k$ in the tree $T$, all of these
nodes are possibly the last surviving node $v$ evoked above: the family of events 
\begin{equation}\label{Ei}
E_i = \{u \textrm{ is not elected and the last neighbor of $u$ was }v_i\}.
\end{equation}
are the ``disjoint events'' mentioned above. We just have to compute $\mathbb{P}(E_i)$.

Our idea to compute the probability of this event is to change of
point of view, and to introduce a notion of \it directed elimination\rm:
if $u$ is eliminated before $v$, this means that the sub-tree
$T[u,\ca{v}]$ -- which is defined to be the tree rooted in $u$ maximal
for the inclusion in $T$ which does not contain $v$ (see Fig. 2)) -- disappears
entirely before $T[v,\ca{u}]$; in the tree $T[u,\ca{v}]$ the
elimination is done from the leaves to the root $u$.

\subsection{Directed elimination in rooted trees}  
\begin{figure}[htbp]
\psfrag{u}{$u$}
\psfrag{v}{$v$}
\centerline{\includegraphics[height=1.8cm]{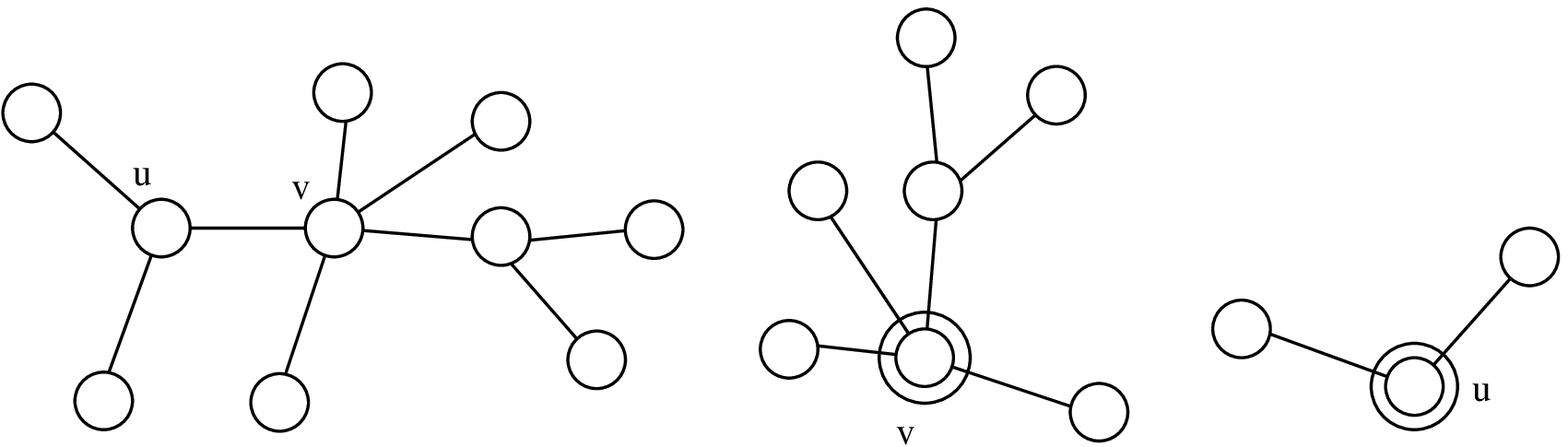}}
\label{el}
\caption{A tree $T$, and the two rooted trees $T[v,\ca{u}]$ and $T[u,\ca{v}]$}
\end{figure}
We define an algorithm $\ALGO^\star(\Delta)$ (very similar to
$\ALGO(\Delta)$) which aims to eliminate all the
nodes of a \textit{rooted} tree, from the leaves to the root. 
We do not investigate the election since the last
living node will be the root, but we are interested in
the duration of the directed elimination of the whole tree. \par

We define $\ALGO^\star(\Delta)$ recursively on a rooted tree
$\tau$. The only difference between $\ALGO(\Delta)$ and
$\ALGO^\star(\Delta)$ is that with $\ALGO^\star(\Delta)$ the root of
$\tau$ is never considered as a leaf: using $\ALGO^\star(\Delta)$ \\
-- the leaves
of $\tau$ are eliminated as with $\ALGO(\Delta)$, transmit and
receive the same information, and compute their remaining lifetimes distribution
with the same function $\Delta$, but \it the root of $\tau$ is not
considered as a leaf\rm, even if it has only one child,\\
-- when the
root $v$ of $\tau$ becomes alone, it has received some information
from its neighbors (or none if it was yet alone at time 0), then it
computes using $\Delta$ the distribution ${\cal D}_v^\star$, and generate $D^\star_v$ accordingly; in other words, the root once alone behaves as a leaf in $\ALGO(\Delta)$. After the delay $D^\star_v$, 
$v$ disappears.\medskip

We define the \it duration \rm $D^\star(\tau)$ of the whole tree $\tau$ rooted in $v$ according to $\ALGO^{\star}(\Delta)$ as the date of disappearance of $v$.
If $\tau$ is a rooted tree with root $u$, and such that the subtree
of $\tau$ rooted at the children of $u$ are $\tau_1,\dots, \tau_k$:
one has
\begin{equation}\label{m2}
D^\star(\tau)=D^\star_u+\max_i D^\star(\tau_i);
\end{equation}
$D^\star_u$ has a distribution given by $\Delta$ with the same rules
as in $\ALGO(\Delta)$. \medskip

We come back in the election problem in a (unrooted) tree $T$ according to $\ALGO(\Delta)$.
Let $u$ and $v$ be two neighbors in a tree $T$; consider in one hand the event
\[E_{u,v}=\{u \textrm{ is not elected and the last neighbor of $u$ is }v\}\]
corresponding to a generic event $E_i$ in (\ref{Ei}). In the other
hand, the two trees $T[u,\ca{v}]$ and $T[v,\ca{u}]$ are rooted trees,
respectively in $u$ and $v$; consider two independent directed
eliminations on these trees as explained above, and denote by
$D^\star(T[u,\ca v])$ and $D^\star(T[v,\ca u])$ their independent
durations. It turns out that
\begin{proposition}The following identity holds true:
\begin{equation}\label{peuv}
\mathbb{P}(E_{u,v}) = \mathbb{P}\big( D^\star(T[u,\ca v])<D^\star(T[v,\ca u])\big).
\end{equation}
\end{proposition}
\begin{proof}
 \rm We propose a proof via a coupling argument. The idea is to compare
 the election process which takes place in $T$ with the directed
 eliminations in $T[v,\ca{u}]$ and $T[u,\ca{v}]$, that are
 directed. The comparison is not immediate since these algorithms 
are not defined on the same probability space.\par 
The algorithms $\ALGO(\Delta)$ and $\ALGO^\star(\Delta)$
 allow each node $u$ to choose a distribution ${\cal D}_u$ or ${\cal
 D}^\star_u$ depending on the information received, from which the
 nodes generate their lifetimes $D_u$ or $D^\star_u$. According to
 Remark \ref{rem1}, a variable $U$ uniform is sufficient to generate
 $D_u$ or $D^\star_u$. Hence, we  suppose that at time 0 each
 node $w$ in the tree $T$ has at its disposal a real number $U_w$
 obtained by a uniform random generator on $[0,1]$.  This is
 the key-point: a node $w$ in $T$ maybe considered also as
 a node in $T[v,\ca{u}]$ or in $T[u,\ca{v}]$, depending on which of
 these trees it belongs. If one now executes  $\ALGO(\Delta)$ on $T$ 
and $\ALGO^\star(\Delta)$ on $T[u,\ca{v}]$ and
 $T[v,\ca{u}]$ using the variable $U_w$ for the generation of the
 $D_w$'s and the $D^\star_w$'s, one can compare the events
 $\{E_{u,v}\}$ and $\{D^\star(T[u,\ca v])<D^\star(T[v,\ca u])\}$, since
 they are now on the same probability space. \par 
It turns out that for each assignment of the
 $U_w$'s, we have
$\{E_{u,v}\}=\{D^\star(T[u,\ca v])<D^\star(T[v,\ca u])\}$. Indeed, since both
 algorithms use the $U_w$'s, since the
 algorithms have the same constructions and the same rules concerning
 $\Delta$, we see that the disappearance of leaves coincide in the two
 models till the disappearance of $u$ or of $v$: after this time, the
 information transmitted are different, and then the two processes
 evolve in a non comparable manner.  Now, in the election process
 $\ALGO(\Delta)$ in $T$, if $u$ is eliminated before $v$, then the
 tree $T[u,\ca v]$ has lived a directed election, and thus
 $D^\star(T[u,\ca v])$ coincides with the disappearance time of $u$
 (for $\ALGO(\Delta)$). At this time, since $v$ is still alive, this
 means that the directed elimination in $T[v,\ca u]$ is not finished,
 thus $D^\star(T[u,\ca v])<D^\star(T[v,\ca u])$. Conversely, if
 $D^\star(T[u,\ca v])<D^\star(T[v,\ca u])$, then $u$ disappears before
 $v$ according to $\ALGO(\Delta)$, since till the time
 $\min(D^\star(T[u,\ca v]),D^\star(T[v,\ca u]))$ the two elimination
 processes coincide. \par We then have construct a probability space
 (the one where are defined the $U_w$'s) on which the two events
 $\{E_{u,v}\}$ and $\{D^\star(T[u,\ca v])<D^\star(T[v,\ca u])\}$
 coincide; thus, they have the same probability.
\end{proof}

As a corollary we have
\begin{corollary}\label{elec-gene}
Let $u$ be a node of a tree $T$ and $u_1$,..., $u_k$ its
neighbors. Using $\ALGO(\Delta)$%The probability that $u$ is eventually elected according
%to the algorithm $\Delta$ is given by
\begin{equation}
q_u= 1-\sum_{1\leq i\leq k}\mathbb{P}\big( D^\star(T[u,\ca
u_i])<D^\star(T[u_i,\ca u])\big).
\end{equation}
\end{corollary}

\section{First category: around the $(\max,+)$ algebra}

In this category, the distribution ${\cal D}_u$  are either the exponential
distribution or a convolution of such distributions. We will see that this category contains the
algorithm of Métivier \&. al. allowing to elect uniformly in the
tree, an algorithm electing proportionally to positive integer
valued prescribed weights, some algorithms allowing to elect
proportionally to some structural features of the tree. \par
Before doing this, we recall some
classical facts. In the sequel ${\mathcal E}^{[a]}$ denote a r.v. having the $\Expo(a)$ distribution, and 
$M_n= \max_{1\leq i \leq n}{\mathcal E}_i^{[1]}$ is the maximum of $n$
i.i.d. r.v. $\Expo(1)$ distributed. The distribution of $M_n$ is
denoted from now on by ${\cal M}_n$ (we have
$\mathbb{P}(M_n\leq x)=(1-\exp(-x))^n$, for any $x\geq 0$).
\begin{lemma}\label{expo}
Let ${\mathcal E}^{[1]},...,{\mathcal E}^{[n]}$ be $n$ independent
exponential random variables with parameters $1,\dots,n$. The random
variables ${\mathcal E}^{[1]}+...+{\mathcal E}^{[n]}$ has
distribution ${\cal M}_n$.
\end{lemma}

\begin{proof} Consider $(\hat{{\mathcal E}_i},\ 1\leq i \leq n),$ the
  order statistics of $n$ i.i.d. $\Expo(1)$ random variables
  ${\mathcal E}^{[1]}_1,\dots,{\mathcal E}^{[1]}_n$, that is the
  sequence $({\mathcal E}_i^{[1]},\ 1\leq i \leq n),$ sorted in the
  increasing order.  The variable $M_n=\max {\mathcal E}_i^{[1]} $ is
  also the sum of the random variables $\hat{\mathcal
  E}_i-\hat{\mathcal E}_{i-1}$, for $i=1,\dots,n$ with the convention
  $\hat{\mathcal E}_{0}=0$. Using the memoryless property of the
  exponential distribution, one has $\hat{{\mathcal E}_i}-\hat{
  {\mathcal E}}_{i-1} \stackrel {d}{=} {\mathcal E}^{[n+1-i]}$ for all
  $i \in\{1,\dots,n\}$, and the variables $(\hat{{\mathcal E}_i}-\hat{
  {\mathcal E}}_{i-1})$ are independent (for more details, see
  Proposition p.19 in Feller \cite{FEL}).
\end{proof}

From the lemma we easily derive:
\begin{corollary}\label{cor-expo}
% $i)$ If $M_{n}$ and ${\mathcal E}^{[n+1]}$ are independent
%$$M_{n+1}\stackrel {d}{=}M_{n}+{\mathcal E}^{[n+1]}.$$  .\\
$i)$ Consider  $k\geq 1$ positive integers $a_1,\dots,a_k$ summing to
$n$. If the r.v.  $M_{a_i}$'s are independent, and independent of
${\mathcal E}^{[n+1]}$ then $
M_{n+1} \stackrel {d}{=} {\mathcal E}^{[n+1]}+\max_{1\leq i \leq k}
M_{a_i}.
$ 
\\ 
$ii)$ For any $k\geq 1$ and $n\geq 1$, set
\begin{equation}
\label{first-type}
Y_{n,k}\stackrel {d}{=}{\mathcal E}^{[n+1]}+ {\mathcal E}^{[n+2]}
+...+{\mathcal E}^{[n+k]},
\end{equation}
where the variables ${\mathcal E}^{[n+i]}$ are independent. We have
$M_{n+k}\stackrel {d}{=} M_{n}+Y_{n,k}.$
\end{corollary}

%\begin{proof}
%$(ii)$ is a direct consequence of Lemma \ref{expo}. For $(i)$, write
%$\max_{1\leq i \leq k} M_{a_i}\stackrel{d}= M_n$ and apply the
%lemma.
%\end{proof}

\subsection{The algorithms of the first category}

The first category of algorithms we design is based on Corollary
\ref{cor-expo}. It may be more easily understood via
the directed elimination $\ALGO^\star(\Delta)$, where the duration of a rooted tree $\tau$
according to $\ALGO^\star(\Delta)$ will have distribution ${\cal M}_n$,
for some $n$. The application $\Delta$ will take its values in the set of
distributions $\{{\cal Y}[n,k],n\geq 1, k\geq 1\}$, where ${\cal
Y}[n,k]$ is the distribution of $Y_{n,k}$ (given in (\ref{first-type})).\par
The only difference between the algorithms of the first category is the computed values $\Gamma_u$'s~: the class of algorithm considered is then simply parametrized by the possible computed values $\Gamma$ satisfying the constraint below. It is convenient to consider bi-dimensional computed values $\Gamma_u=(C_u,g_u)$ where $C_u$ will be use to add some quantities coming from the received information, and $g_u$ is used to make some local computations.  \par
 Here are in two points the description of all the
algorithms of the first category:\\
-- At time $0$, the computed value $\Gamma_u$ of any leaf $u$ is $\Gamma_u=(0,g_u)$ where $g_u$ is a positive integer. Then set
\begin{equation}\label{leaf}
{\cal D}_u = {\cal Y}[0,g_u]\stackrel{d}{=}M_{C_u+g_u}.
\end{equation}
-- Let $u$ be an internal node in $T$ becoming a leaf; let $f$
be the received information, and in particular let
$\Gamma_1=(C_1,g_1),\dots,\Gamma_k=(C_k,g_k)$ be the computed values of its eliminated neighbors. Then the node $u$ compute an integer value $g_u$ according to its information ($f$ and $L^\star_u$), and let 
$C_u=\sum_{i=1}^k C_i+g_i.$
Then set ${\cal D}_u= {\cal Y}\left[C_u,g_u\right].$

%The computed value of $u$, that will be transmitted after the
%elimination of $u$ is given by
%\begin{equation}\label{transmit}
%\Gamma_u=g_u+\sum_{i}\Gamma_i.
%\end{equation}
\medskip

\noindent Let us think in terms of directed elimination. Recall that 
the notion of computed values are defined similarly in $\ALGO^\star(\Delta)$ 
and in $\ALGO(\Delta)$, but in the directed case, it is convenient to
make appear the tree notation in the computed values instead of the
node notation. \par 
If a rooted tree $\tau$ is reduced to a leaf $u$,
set $C(\tau)=0, g(\tau)=g_u$. If $\tau$ has root $u$, and if the
sub-trees rooted at the children of $u$ are $\tau_1,\dots,\tau_k$,
then set $C(\tau)=\sum_{i=1}^k C(\tau_i)+g(\tau_i).$
The lifetime of the root of $\tau$ is then
distributed as the maximum of the $D^\star(\tau_i)'s$ plus a random
variable distributed as ${\cal Y}(C(\tau),g(\tau))$.

To simplify a bit the formula, for any rooted tree $\tau$, let
\begin{equation}
\Theta(\tau)=g(\tau)+C(\tau).
\end{equation}

\begin{proposition}\label{elec-first}
For any algorithm $\ALGO^\star(\Delta)$ of the first category 
the duration of  a rooted tree $\tau$ satisfies 
\[D^\star(\tau)\stackrel{d}{=} M_{\Theta(\tau)}.\]
\end{proposition}
\begin{proof}
%This result will appear to be a consequence of (\ref{m2}) and Corollary \ref{cor-expo}. 
The
lifetime of a tree $\tau$ reduced to a leaf is ${\cal Y}(0,g(\tau))=M_{C(\tau)+g(\tau)}=M_{\Theta(\tau)}$. Assume by
induction that the proposition is true for any rooted tree having
less than $n$ nodes. Consider now $\tau$ a rooted tree with $n$ nodes and the $\tau_i$ defined as above. 
By recurrence
$D^\star(\tau_i)\stackrel{d}=M_{\Theta(\tau_i)}$, and thus, by independence of
the $M_{\Theta(\tau_i)}$'s, $D^\star(\tau)={\cal
Y}\left[\sum_{i}\Theta(\tau_i),g(\tau)\right]+\max_i M_{ \Theta(\tau_i)}$ is in
distribution equal to $M_{(\sum_{i}\Theta(\tau_i))+g(\tau)}\stackrel{d}{=}M_{\Theta(\tau)}$ by Corollary \ref{cor-expo}.
\end{proof}

As a corollary we have
\begin{theorem}For any algorithm $\ALGO(\Delta)$ of the first category, any tree $T$, 
%the probability of election of a node $u$ of $T$ is given by 
\begin{equation}\label{gam}
q_u= 1-\sum_{1\leq i\leq k}\frac{\Theta(T[u_i,\ca u])}{\Theta(T[u,\ca
u_i])+\Theta(T[u_i,\ca u])}
\end{equation}
\end{theorem}
\begin{proof}
This is a consequence of Propositions \ref{elec-gene} and
\ref{elec-first} and of the following identity: if $M_a$ and $M_b$ are
independent, then $\mathbb{P}(M_a<M_b)=a/(a+b)$.
\end{proof}

This theorem has a direct consequence quite surprising, since it
deals with very general function $\Gamma$. It is obtained by summing
Equality (\ref{gam}) over all nodes:
\begin{corollary} For {\it any} tree $T$, any choice of positive
  integer values function $\Gamma_u=(C_u,g_u)$
$$\sum_u\left[1-\sum_i\frac{\Theta(T[u_i,\ca{u}])}{\Theta(T[u_i,\ca{u}])+\Theta(T[u,\ca{u_i}])}\right]=1.$$
\end{corollary}
Remark \ref{fr} ensures that almost surely the election eventually
succeeds. Indeed, each leaf eventually dies
out with probability one, and then the election stops after a finite time. 
All the disappearance dates are different, since the lifetimes distributions have no atom: at the end it
eventually remains only one leaving node which is elected.

\begin{remark}
%\bf Remark : \rm 
In general the denominator in the RHS of
(\ref{gam}) depends on the node $u$ and, thus, apart from the two
first examples below where this denominator is constant, the formula
(\ref{gam}) cannot be ``simplified''.
\end{remark}
\subsection{Examples }\label{exa}
\begin{enumerate}
\item The uniform electing algorithm (treated in Example \ref{esa1})
  is a particular case of this model by letting $g_u=1$ and,
  therefore, $\Theta(t)=|t|$, the total number of nodes in $t$. Since
  each node is either in $T[u,\ca u_i]$ or in $T[u_i,\ca u]$, by
  (\ref{gam})
\begin{eqnarray*}
q_u&=&1-\sum_{1\leq i\leq k} \frac{|T[u_i,\ca u]|}{|T[u,\ca
u_i]|+|T[u_i,\ca u]|}=1-\frac{|T \setminus\{u\}|}{|T|}=\frac{1}{|T|};
\end{eqnarray*}
this is the uniform distribution on $T$, as found by M\'etivier \& al.
\item Assume that all prescribed weights are positive integers. If
  $g_u=w_u$ for every nodes then $\Theta(t)=\sum_{u \in t} w_u$ the
  total weight of the rooted tree $t$. In this case $q_u=
  \frac{w_u}{w(T)}$ where $w(T)=\sum_{u\in T} w(u)$ is the total weight
  in $T$. Indeed, in the RHS of (\ref{gam}) the denominator is equal
  to $w(T)$ whatever is the value of $i$, and summing the numerators
  gives $w(T)-w_u$.
\item For $g_u=\deg(u)$,  $q_u$ becomes proportional to $\deg(u)$ (take  $w_u=\deg(u)$ in the previous point 2).
\item In the case where $g_u=1$ for the leaves and $g_u=|t|$ more
  generally for all the nodes, then $\Theta(t)=PLS(t)+|t|$ becomes the
  path length of (the rooted tree) $t$ plus its size. Then Formula
  (\ref{gam}) gives the value of $q_u$.
\end{enumerate}

\subsection{Real-valued weights}
\label{treo}
In Example \ref{exa}.2, we gave an algorithm of the
first category such that $q_u$ is proportional to $w_u$ provided that the
$w_u's$ are integers. The
computations relying on Corollary \ref{cor-expo}, the weights have
to be integer valued, or say have a known common divisor. 
A natural question arises: is there an algorithm such that
$q_u$ is proportional to general real-valued weights $w_u$'s? We were not able to answer
to this question, but using a randomized version
of the algorithms of the first category, we provide an
algorithm that may fail with a small probability, but such that
conditionally on success, the $q_u$'s are indeed proportional to the
$w_u$'s.

\par The difference with the algorithm described above is as
follows. Instead of using its weight $w_u$ as a parameter in a
distribution ${\cal Y}(n,k)$, a node $u$ becoming a leaf, uses its
weight $w_u$ as a parameter of a Poisson distribution: it generates
$W_u$ a r.v. following the Poisson($w_u$) distribution and then uses
this integer as its weight in the description of algorithms of the
first category we gave. In other words, the computed value $g_u$
instead of being simply $w_u$ will take the value
$k$ with probability $\exp(-w_u)w_u^k/k!$. \\
%Finally $u$, conditionally to $\Gamma_u$ generates
%its lifetime according to ${\cal M}_{\Gamma_u}$.  \par But before
%going into the proof, 
Let us discuss some points linked to the failure
of the algorithm.
\begin{remark}
-- If the random generated $W_u$ is zero for some $u$,
then conditionally to $W_u$ the remaining lifetime is $\Expo(0)$
distributed, that is zero almost surely: $u$ is
eliminated immediately.\\
-- If {\it all} nodes generate zero, then the algorithm fails:
  {\it it terminates without choosing any node}.
  The probability of failure for the algorithm is 
  $e^{-w(T)}$ where  $w(T)=\sum_{u\in V}w_u$ is the total weight. It becomes
  insignificant whenever $w(T)$ grows. To guarantee the success with a
  high probability, it suffices to multiply $w$ by a great number $c$
  known by all nodes.
\end{remark}
The following lemma, which is easily proved, simplifies the proof of the
main proposition of this section.
\begin{lemma}
Let $X_1,...,X_n$ be $n$ independent r.v. of Poisson distributions
  with parameters ${\lambda}_1,...,{\lambda}_n$ respectively. %For a
%  given integer $k>0$, and $l\in\{0,\dots,k\}$ we have:
%$${\mathbb P}\left(X_1=l~\Big|\ \sum_{1\leq i \leq n}X_i=k\right)={k \choose
%  l}\left(\frac{{\lambda}_1}{\sum_i {\lambda}_i}\right)^l
%\left(\frac{\sum_{i\neq 1}{\lambda}_i}{\sum_i {\lambda}_i}\right)^{k-l},
%$$
%i.e. t
For any $k>0$, the distribution of $X_1$ conditionally on
$X_1+\dots+X_n=k$ is binomial
$B(k,\lambda_1/(\lambda_1+\dots+\lambda_n))$.
\end{lemma}
\begin{proposition}\label{weighted}
Let $T$ be any tree. The probability that the algorithm chooses a
node $u$ conditioned by the event that not all nodes generate 0
is proportional to $w_u$~:
%\begin{eqnarray*} 
${\mathbb P}\left(u \mbox{ elected }\ \Big| \ \sum_{v \in
V}W_v>0\right)=  {w_u}/{w(T)}.$
%\end{eqnarray*}
\end{proposition}
\begin{proof}
Consider some integers $(k_v)_{v\in V}$ , with at least one $k_v>0$.
Given the values $W_v=k_v$ according to Section \ref{exa},
second example, we have:
\begin{eqnarray*}
{\mathbb P}\left(u \mbox{ elected } | ~W_v=k_v \mbox{ for any } v
\mbox{ in }T \right)&=& {k_u}/({\sum_{v\in V} k_v}).
\end{eqnarray*} 
Therefore the probability that the algorithm chooses $u$ conditioned
by $\sum_vW_v>0$, is nothing but:
\begin{eqnarray*}
{\mathbb P}(u \mbox{ elected }\Big| \ \sum_vW_v>0)&=&{\mathbb
  E}\left(\frac{W_u}{\sum_vW_v}\ \Big|\ \sum_vW_v>0
  \right),
\end{eqnarray*}
where ${\mathbb E}$ denotes the expected value. But then, according to
the previous lemma, for a fixed $k>0$,
$$ {\mathbb
  E}\left(\frac{W_u}{\sum_vW_v}\ | \sum_vW_v=k
\right)=  \frac {w_u}{\sum_v w_v}.
$$ This implies that if the sum of generated numbers is positive,
whatever the values it takes, the probability of $u$ to be elected is
$\frac {w_u}{\sum_v w_v}$. The proposition follows.
\end{proof}

\section{Second category: around the stable distributions}
The second category relies on
Formula (\ref{m2}). One sees that choosing a suitable $D^\star$ may
let the $\max$ operator acting on the RHS disappears: the idea is to choose $D_u^\star$ under the form
\begin{equation}\label{zea}
D_u=X^u-\max_i D(\tau_i)+\sum_i D(\tau_i)
\end{equation} 
for some $X^u$ whose distribution depends of the information received
by $u$. In this case Formula (\ref{m2}) concerning the directed
elimination becomes simply
\[D^\star(\tau)=X^u+\sum_i D^\star(\tau_i).\]
And the duration of a rooted tree satisfies:
\begin{eqnarray}\label{zeb}
D^\star(\tau)%&=& %D^\star_u+\max_i\{D^\star(t_i)\}\\
&=&X^{u}+\sum_i D^\star(\tau_i)=  \sum_{v\textrm{ nodes in }\tau} X^v.
\end{eqnarray}
Once again, the involved variables $X^v$  have a distribution that may
depend on the history of the elimination of the sub-tree of $\tau$
rooted in $v$. The algorithms of the second category are parametrized
by all the possible distribution for $X^u$ (the variables $X^u$
appearing in (\ref{zea}) and (\ref{zeb})). \par

In the case where the $X^v$ are i.i.d, the distribution  of $D^\star(\tau)$ is simple: it
is a sum of $|\tau|$ i.i.d. random variables, and then it is indexed
by the unique integer $|\tau|$.  Denoting by $S_n$
a sum of $n$ i.i.d. copies of $X^v$, according to Corollary
\ref{elec-gene} we have for a node $u$ having $u_1,\dots,u_k$ as neighbors,
\begin{equation}\label{un-peu-moche}
q_u=1-\sum_{1\leq i\leq k}\mathbb{P}\big( S_{|T[u,\ca u_i]|}<S_{|T[u_i,\ca u]|}\big).
\end{equation}
There is an interesting case where the computation in
(\ref{un-peu-moche}) can be made explicitly, and leads to close
formulas: the case of the stable distribution with index $1/2$. The
stable distributions are the families of distribution that are stable
for the convolution (see Feller \cite{FEL} for more
information).  We say that $X$ has the stable
distribution with index $1/2$ if the density of $X$ is $f(t)={\bf 1}_{t\geq 0} \frac{e^{-1/(2t)}}{\sqrt{2\pi t^3}}.$
If $X_1,\dots,X_k$ are independent copies of $X$ then
$S_k=X_1+\dots+X_k\stackrel {d}{=}k^2 X$.  
%The stable distribution with index $1/2$ appears in probability theory as
%the distribution of the hitting time of 1 by the Brownian motion; the
%Markovian properties of the Brownian motion, as well as its scaling
%properties, say that the hitting time of $k$ is distributed as $k^2$
%times the hitting time of $1$, and is the sum of $k$ copies of $X_1$,
%by decomposition of the Brownian motion between the hitting times of
%$1,2,\dots, k$. 
Consider now $S_m$ and $S'_n$ two independent sums
of $m$ and $n$ independent copies of $X$.
One has
\begin{equation}
\mathbb{P}(S_m< S'_n)=\mathbb{P}(m^2 X \leq n^2 X')
\end{equation}
for two copies $X$ and $X'$ of $X$. 
Using the density of $X$ and $X'$, one gets
$\mathbb{P}(S_m< S'_n)=\frac{2}{\pi}\arctan(n/m).$
Hence
\begin{lemma}\label{arc}
For any tree $T$, for any node $u$ having $u_1,\dots,u_k$ as
neighbors, under the algorithm presented above
\[q_u=1-\sum_{1\leq i\leq k}\frac2\pi \arctan\left(\frac{|T[u_i,\ca u]|}{|T[u,\ca u_i]|}\right). \]
\end{lemma}
In particular, since $\sum q_u=1$ this gives for each tree a formula
related to the arctan function. We review below some examples and
derive formulas.

\subsection{Applications: some identities involving the $\arctan$ function}
Consider the star tree with $n$ nodes: it is the tree where a node $v$
has $n-1$ neighbors, say $v_1,\dots,v_{n-1}$.  By symmetry $q_{v_i}$
does not depend on $i$; since $v_i$ has for only neighbor $v$, by
Lemma \ref{arc}
\[q_{v_1}=1-({2}/{\pi})\arctan(n-1).\]
Using again Lemma \ref{arc},
one has for the center of the star tree
\[q_v=1-\frac{2(n-1)}{\pi}\arctan\left(\frac{1}{n-1}\right).\]
Since $q_v+\sum_{i=1}^{n-1}q_{v_i}=1$ (since a node is eventually
elected with probability 1), we get for any $n\geq 2$,
\begin{equation}\label{ar-s}
\arctan(n-1)+\arctan(1/(n-1))=\pi/2.
\end{equation}
Consider now a sequence of trees $T_n$ such that $T_n$ is formed by two
stars having $\alpha_n+1$ and $\beta_n+1$ nodes with center $u$ and
$v$, linked by an edge between $u$ and $v$. The election probability
of any leaf is $q_{v_i}=1-({2}/{\pi})\arctan\left(\alpha_n+\beta_{n+1}\right),$
when
\begin{eqnarray*}
q_u&=&1-\frac{2\alpha_n}{\pi}\arctan\left(\frac 1 {\alpha_n+\beta_{n+1}}\right)-\frac{2}{\pi}\arctan\left(\frac{\beta_{n}+1}{\alpha_n+1}\right)\\
q_v&=&1-\frac{2\beta_n}{\pi}\arctan\left(\frac1{\alpha_n+\beta_{n+1}}\right)-\frac{2}{\pi}\arctan\left(\frac{\alpha_{n}+1}{\beta_n+1}\right).
\end{eqnarray*}
Using $(\alpha_n+\beta_n)q_{v_1}+q_u+q_v=1$ and (\ref{ar-s}), we get 
\[\frac{2}{\pi}\left(\arctan\left(\frac{\alpha_n+1}{\beta_n+1}\right)+\arctan\left(\frac{\beta_n+1}{\alpha_n+1}\right)\right)=1.\]
If $\alpha_n/\beta_n\to x >0$, by continuity of $\arctan$ one obtains the famous
formula
\[\arctan(x)+\arctan(1/x)={\pi}/{2}.\]
Going further, let $T_n$ be the sequence of trees having a path of size $k$
($k$ nodes $u_1,\dots,u_k$ such that there is an edge between $u_i$
and $u_{i+1}$ and such that $u_i$ has $\alpha_{n,i}$ other neighbors
that are leaves).  The probability of election of any of the $\sum
\alpha_{n,i}$ leaves is
$q_l=1-\frac{2}{\pi}\arctan(\sum \alpha_{n_i}+k-1)$, that of $u_i$
is
\[1-\frac{2}{\pi}\left[\alpha_{n,i}\arctan\left(\frac1{\sum \alpha_{n,i}+k-1}\right)+\arctan\left(\frac{\sum_{j>i}(\alpha_{n,j}+1)}{\sum_{j\leq i}(\alpha_{n,j}+1)}\right)+\arctan\left(\frac{\sum_{j<i}(\alpha_{n,j}+1)}{\sum_{j\geq i}(\alpha_{n,j}+1)}\right)\right].\]
Finally, assuming that for any $i$, $\alpha_{n,i}\to \alpha_i$ for
some positive real number $\alpha_i$, we get by continuity, and
using that the sum of all events must be 1, that for any positive real
number $\alpha_1,\dots,\alpha_k$,
\begin{equation}\label{form2}
\sum_i\left[\arctan\left(\frac{\sum_{j>i}\alpha_{j}}{\sum_{j\leq
i}\alpha_j}\right)+\arctan\left(\frac{\sum_{j<i}\alpha_{j}}{\sum_{j\geq
i}\alpha_{j}}\right)\right]=\frac{\pi}{2}(k-1).
\end{equation}
%% $\bullet$ Taking now a node $u$ having $k$ neighbors having
%% respectively $\alpha_{n,1},\dots,\alpha_{n,k}$ neighbors, by the same
%% methods, one gets
%% \begin{equation}\label{form1}
%% \sum_i\left[\arctan(\frac{\sum_{j\neq
%% i}\alpha_{j}}{\alpha_i})\right]=\pi/2.
%% \end{equation}
Each simple finite tree used as a skeleton on which are grafted some
packets of leaves (with size $\alpha_{n,k}$, $k$ corresponding to a
labeling of the nodes of the skeleton) will provide a formula similar
to (\ref{form2}).

\small\renewcommand{\baselinestretch}{1}

\bibliographystyle{plain}
\bibliography{elec}

\end{document}